\documentclass[conference]{IEEEtran}
\IEEEoverridecommandlockouts
\usepackage{cite}
\usepackage{amsmath,amssymb,amsfonts}

\usepackage{graphicx}
\usepackage{textcomp}
\usepackage[T1]{fontenc}
\usepackage[utf8]{inputenc}
\usepackage{lmodern}
\usepackage{hyperref}
\usepackage[english]{babel}
\usepackage{todonotes}
\usepackage{algorithm}
\usepackage{algpseudocode}
\usepackage{mathtools}
\usepackage{mathdots}
\usepackage{tikz-cd} 
\usepackage{wrapfig}
\usepackage{float}
\usepackage{quiver}
\usepackage{adjustbox}
\usepackage{stackengine,scalerel}
\newcommand\eye{\ensurestackMath{\stackinset{c}{}{c}{-.33pt}%
		{\bullet}{\bigcirc}}}

\makeatletter
\newcommand{\removelatexerror}{\let\@latex@error\@gobble}
\makeatother
\newcounter{cases}
\newcounter{subcases}[cases]
\newenvironment{mycase}
{
	\setcounter{cases}{0}
	\setcounter{subcases}{0}
	\newcommand{\case}
	{
		\par\indent\stepcounter{cases}\textbf{Case \thecases.}
	}
	
}
{
	\par
}
\renewcommand*\thecases{\arabic{cases}}

\DeclareMathOperator*{\argmin}{argmin}

\usepackage{amsthm}
\usepackage{tikz}
\usetikzlibrary{cd}
\usetikzlibrary{petri,positioning}
\newtheorem{theorem}{Theorem}
\newtheorem{lemma}[theorem]{Lemma}

\newtheorem*{lemma*}{Lemma}
\newtheorem*{theorem*}{Theorem}
\newtheorem{definition}{Definition}

\newtheorem{example}{Example}

\def\BibTeX{{\rm B\kern-.05em{\sc i\kern-.025em b}\kern-.08em
		T\kern-.1667em\lower.7ex\hbox{E}\kern-.125emX}}
\begin{document}
	\title{Timed Alignments with Mixed Moves
		\thanks{Neha Rino was funded by the International Master's Scholarships Program IDEX of Université Paris-Saclay.}
}

\author{
	\IEEEauthorblockN{Neha Rino}
	\IEEEauthorblockA{\textit{LMF, ENS Paris-Saclay,}\\
		\textit{CNRS, Université Paris-Saclay, Inria} \\
		Gif-sur-Yvette, France\\
		\url{neha.rino@ens-paris-saclay.fr}}
	\and
	\IEEEauthorblockN{Thomas Chatain}
	\IEEEauthorblockA{\textit{LMF, ENS Paris-Saclay,}\\
		\textit{CNRS, Université Paris-Saclay, Inria} \\
		Gif-sur-Yvette, France\\
		\url{thomas.chatain@ens-paris-saclay.fr}\\
		ORCID 0000-0002-1470-5074}
}

	\maketitle
	
	\begin{abstract}
		The subject of this paper is to study conformance checking for timed models, that is, process models that consider both the sequence of events in a process as well as the timestamps at which each event is recorded. 
		Time-aware process mining is a growing subfield of research, and as tools that seek to discover timing related properties in processes develop, so does the need for conformance checking techniques that can tackle time constraints and provide insightful quality measures for time-aware process models. 
		In particular, one of the most useful conformance artefacts is the alignment, that is, finding the minimal changes necessary to correct a new observation to conform to a process model.  
		This paper follows a previous one, where we have set our problem of timed alignment. In the present paper, we solve the case where the metrics used to compare timed processes allows mixed moves, i.e. an error on the timestamp of an event may or may not have propagated to its successors, and provide linear time algorithms for distance computation and alignment on models with sequential causal processes. 
	\end{abstract}
	
	\begin{IEEEkeywords}
		Conformance checking, Alignments, Timestamps, Time Petri nets
	\end{IEEEkeywords}
	
	\section{Introduction} 
	
	\subsection{Conformance Checking and Alignments} 
	
	Process mining studies vast systems through their event logs, and seeks to extract meaningful ways to model the underlying patterns or processes that govern the behaviour of the system in order to better understand it, or predict future behaviour \cite{A16}.
	Once such a process model is obtained, it is natural to ask how one is sure the obtained model is a reasonable approximation of the system’s behaviour at all, especially given the lack of explainability in the blackbox approach ML takes to producing solutions.
	This is where conformance checking comes into the picture, as it is the art of judging the performance of a process model by relating modelled and observed behaviour of a process to each other, without depending on the origin of the model \cite{CDSW18}.
	Observed behaviour comes in the form of traces in an event log, as a sequence of events occurring during the functioning of the system, while process models are blueprints that describe what the underlying processes of any given system are supposed to look like.
	We often do not want the system to always precisely generate any and all possible future system behaviour, but neither should they simply regurgitate the event log and accept no new system behaviours. 
	What is much more useful is a process model that can, up to some small error factor, approximate any reasonable future system behaviour. 
	
	In Arya Adriansyah's seminal thesis \cite{Adriansyah2014AligningOA}, we obtain the notion of an alignment, that is, the minimal series of corrections needed to transform an observed event trace into the execution of the process model that most closely mimics it.  
	This is given as a series of edits, usually insertions or deletions, that transform the observed trace into a process trace.  
	Alignments thereby help pinpoint exactly where inevitable deviations from expected behaviour occur, and the more distant the aligning word of a model is to its observed trace, the worse the model is at reflecting real system behaviour. 
	
	In the untimed case, process models can be represented using a variety of formal objects, such as Petri nets. 
	Assuming the event logs are a list of words over a finite alphabet (the set of possible discrete events), the problem of calculating the alignment has been extensively studied \cite{Adriansyah2014AligningOA} \cite{BCC21}. 
	The notion of distance used on these words used is usually either Hamming distance or Levenshtein's edit distance. 
	It is natural to want to study explicitly timed systems,
	as by considering events along with their timestamps when mining processes, we can glean information about the minimum delay between two events, the maximum duration the system takes to converge upon a state, or check deadlines, all of which are highly relevant in real world applications \cite{Cheikhrouhou2014TheTP} \cite{Eder1999TimeCI} \cite{inbook}. 
	Time-aware process mining seeks to study both what sort of underlying processes govern system behaviour, and what sort of time constraints they can impose on when certain events can occur \cite{RMAW13} \cite{CRHA20}  \cite{AS21}. 
	In addition, one may want to predict the timestamps of processes \cite{SSA11}. 
	In the process mining community, there are ways to use existing process model notation in order to denote time constraints. BPMN 2.0 comes equipped with \emph{timer events} and can record absolute, relative, and cyclical time constraints. 
	For our purposes, we use time Petri nets, an extension of Petri nets equipped with the ability to express constraints on the duration of time between an action being enabled, and its actual occurrence. In particular, in this paper, we restrict ourselves time Petri nets with no branching points. 
	
	As time-aware process mining grows popular, new quality measures and conformance checking techniques must be developed that are sensitive to temporal constraints, but so far in the study of alignments as a conformance checking artefact, we notice that the process model used is never time-aware. 
	This paper seeks to investigate one of the distance functions, $d_N$, defined in \cite{CR22} and solve the alignment problem for the same. 
	
	\subsection{Timed Alignments}
	
	\begin{example} Consider a model of composing and sending an email, where the intervals signify allowed durations between the current event and its immediate predecessor, depicted below with input and output places marked with an $i$ and an $o$ : 
		
		\adjustbox{scale = 0.9}
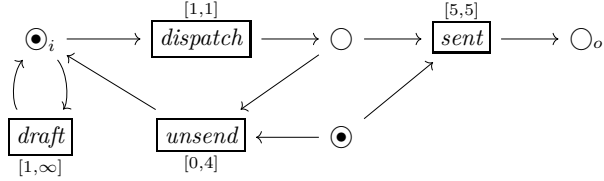
{
			\begin{tikzcd}
				\eye_i & {\overset{[1,1]}{\fbox{dispatch}}} & \bigcirc & {\overset{[5, 5]}{\fbox{sent}}} & \bigcirc_o \\
				{\underset{[1, \infty]}{\fbox{draft}}} & {\underset{[0,4]}{\fbox{unsend}}} & \eye
				\arrow[curve={height=-12pt}, from=2-1, to=1-1]
				\arrow[curve={height=-12pt}, from=1-1, to=2-1]
				\arrow[from=1-1, to=1-2]
				\arrow[from=1-2, to=1-3]
				\arrow[from=1-3, to=1-4]
				\arrow[from=1-3, to=2-2]
				\arrow[from=2-3, to=2-2]
				\arrow[from=2-3, to=1-4]
				\arrow[from=1-4, to=1-5]
				\arrow[from=2-2, to=1-1]
			\end{tikzcd}
		}\\
		One process trace is $(draft, 5)(dispatch, 6)(sent, 11)$, which depicts drafting the message in 5 units and then dispatching it, and having it send successfully. An example of an observed trace which does not conform to the process above would be $(draft, 3)(draft, 5)(dispatch, 7)(dispatch, 7)(sent, 12)$. Clearly, there is an extraneous letter here, as the same message cannot be dispatched twice, and the timestamps for dispatch and sending come too late. One possible optimally close process trace for this observation is $(draft, 3)(draft,6)(dispatch, 7)(sent, 12)$, which deletes the extra $dispatch$, and extends the second $draft$ event by a unit, allowing the rest of the trace to thereby be in time. 
	\end{example}
	
	In \cite{CR22} we posed the general alignment problem, which formalised the notion of aligning timed traces to timed process models.  
	
	\begin{definition}[The General Alignment Problem]
		Given a process model $N$ denoted by a Time Petri Net and an observed timed trace $\sigma$ we wish to find a timed word $\gamma \in \mathcal{L}(N)$ such that $d(\sigma, \gamma) = \min_{x \in \mathcal{L}(N)} d(\sigma, x)$ for some distance function $d$ on timed words. 
	\end{definition}
	
	In this paper, we restrict ourselves to the study of the alignment problem on a simpler model for time-aware processes, ones that we define as sequential process models. These essentially lose the ability to express parallel events the way time Petri nets can, but retain the ability to reason about time constraints on the durations between any two consecutive events. In addition, we focus on the mixed moves distance $d_N$ defined in \cite{CR22}, that is, a distance where edits can either stay local to a given event or propagate forward into the event's causal future. The last assumption we make is that the untimed part of the observed trace conforms to the process, i.e, the only part that requires aligning is the timestamps. This allows us to focus on the timing aspect of the problem, although once this is solved we can easily adapt existing untimed alignment methods like those described in \cite{Adriansyah2014AligningOA} or \cite{BCC21} to align words with both action labels and timestamps that need editing. This brings us to the following formal problem : 
	
	\begin{definition}[The Purely Timed Alignment Problem for Sequential Process Models]
		Given a sequential process model $N$ denoted by a Time Petri Net and an observed timed trace $(w, \sigma)$, such that $w \in Untime(\mathcal{L}(N))$ we wish to find a valid timestamp sequence $\gamma$ such that $d(\sigma, \gamma) = \min_{x \in \mathcal{L}(N)} d(\sigma, x)$.
	\end{definition}
	
	We informally look at an instance of the purely timed problem below : 
	
	\begin{example}
	Consider the process model $N_1$:  
	\adjustbox{scale = 0.9}
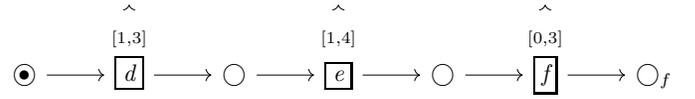
{
		\begin{tikzcd}
			&{}&& {}&&{} & \\
			\eye &\fbox{d}& \bigcirc & {\fbox{e}} & \bigcirc & \fbox{f} & \bigcirc_f\\
			\arrow[from=2-1, to=2-2]
			\arrow[ from=2-2, to=2-3]
			\arrow[from=2-3, to=2-4]
			\arrow[from=2-4, to=2-5]
			\arrow[from=2-5, to=2-6]
			\arrow[from=2-6, to=2-7]
			\arrow["{[1,3]}"{description, pos=0.2}, draw=white, from=2-2, to=1-2]
			\arrow["{[1,4]}"{description, pos=0.3}, draw=white, from=2-4, to=1-4]
			\arrow["{[0, 3]}"{description, pos=0.2}, draw=white, from=2-6, to=1-6]
		\end{tikzcd}
	}
Now, say we observed the trace $(u_1, \sigma_1) = (d, 4)(e, 6)(f, 6)$ that did not fit the model, and we wished to analyse how best to modify them to fit them back into the model. We note that the untimed part does match the model's specification, and in order to make the timestamps fit we can simply edit the timestamp of $d$ to give the following process trace $(v_1, \gamma_1) = (d, 4)(e, 6)(f, 6)$. 
	
	
	\end{example}

	In sections \ref{2} and \ref{3} of this paper we will focus on the problem of computing the distance function $d_N$ and in section \ref{4}, we return to the problem of aligning sequential process models to observed traces. We present linear time algorithms for both computing $d_N$ and solving the alignment problem in this setting, but once again only for linear timed words, and sequential process models.  
	
	\section{Preliminaries : Edit Moves for Timed Words}\label{2}
	
	Much like Levenshtein's edit distance, popularly used in the untimed case of the alignment problem, we view the definition of these distances as an exercise in cost minimisation over the set of all transformations between two words. 
	In order to formalise the same, we define what the valid moves of such a transformation could be. 
	We define \textit{moves} as functions that map one time sequence to another. 
	What sort of functions are useful notions of transformation on a timed system? 

	\begin{example}
		Say we sought to align the timed trace $(4,6,6)$ with $(3, 6, 6)$. 
		It feels reasonable to say the distance between the two is 1. 
		A way to arrive at this conclusion is that if the first timestamp  were just shifted to fire at 3 instead, the whole process would match. This sort of local, almost typographical error can often happen in systems, and it is the simplest kind to fix.  
		
		Now consider aligning the timed traces $(4, 8, 11)$ and $(3, 7, 10)$. 
		Now, when trying to compare these firing sequences, we can view it just like we did for the previous pair, as all of the timestamps firing later than they should, and so moving each timestamp back once, giving an aligning cost of 3. There is however, another way to see this deviation. This cascading chain of errors can be fixed if the first event is moved back to fire at 4, and all the relative relationships between it and its successors are preserved. This views the second and third tasks as only caring about when the first ended, which makes sense, because they are causally linked and hence do only start once their predecessor ends. This means, the switch from the timestamp series $(4, 8, 11)$ to $(3, 7, 10)$ can be viewed as only a cost 1 edit in this sense. This is a slightly more complex error to conceive of, but it is of natural practical use, as if a delay at the beginning caused the whole process to occur too late, the important thing to fix is just the delay at the beginning, and one can sometimes assume with that sorted, the rest of the process will now conform to the model as needed. 
	\end{example}
	
	Based on the above example, we naturally arrived at two types of moves in \cite{CR22}.  
	
	We define a \textit{stamp move} as a move that translates the timing function only at a point, i.e., that edits a particular element of the timestamp series $\tau$. 
	Calculating $d_t$ between two traces can be viewed as a cost minimization process in aligning the traces using only stamp moves. 
	
	\begin{definition}[Stamp Move]
		Given a timing function $\gamma : \{1, \dots n\} \to \mathbb{R}$, formally, we define this as : 
		
		$\forall x \in \mathbb{R}, i \leq n : stamp(\gamma, x, i) =\gamma' $ where
		$$\forall i \leq n : \gamma'(j)  = \begin{cases}
			\gamma(j) + x & j = i \\
			\gamma(j)& otherwise
		\end{cases}$$
		
	\end{definition}
	
	The next type of move we describe is the more novel and interesting \textit{delay move}. 
	By formulating this type of edit move, we seek to leverage the structure of the process model itself, by reflecting the causal relationships between events. 
	
	A stamp move is purely local, in that when a stamp edit occurs at $e$ its immediate successor events shift their relationship with $e$, thereby ensuring that the change in $e$'s timestamp does not derail the rest of the system. 
	On the other hand, any delay move at $e$ will preserve relative relationships in the future, at the cost of shifting the timestamp of every causal descendent of $e$ by the same amount. 
	
	This brings us to the following definition : 
	
	\begin{definition}[Delay Move]
		Given a timing function $\gamma : \{1, \dots n\} \to \mathbb{R}$, formally, we define this as : 
		
		$\forall x \in \mathbb{R}, i \leq n : stamp(\gamma, x, i) =\gamma' $ where
		$$\forall i \leq n : \gamma'(j)  = \begin{cases}
			\gamma(j) + x & j \geq  i \\
			\gamma(j)& otherwise
		\end{cases}$$
		
	\end{definition}
	
The cost of a move is the magnitude $|x|$ in the above definitions, and the cost of a sequence of moves is the sum of the costs of the moves. Armed with these definitions, three natural notions of distance can be constructed. 
	
	\begin{definition}[Stamp Only Distance : $d_t$]
		Given any two timing functions $\tau_1, \tau_2$ over the same causal process $(CN, p)$, we define the stamp-only distance $d_t$ as follows : 
		\small	$$d_t(\tau_1, \tau_2) = \min \{cost(m) | {m \in Stamp^*}, {m(\tau_1) = \tau_2}\}$$
	\end{definition}
	
	\begin{definition}[Delay Only Distance : $d_\theta$]
		Given any two timing functions $\tau_1, \tau_2$ over the same causal process $(CN, p)$, we define the delay-only distance $d_\theta$ as follows : 
		\small	$$d_\theta(\tau_1, \tau_2) = \min \{cost(m) | {m \in Delay^*}, {m(\tau_1) = \tau_2}\}$$
	\end{definition}
	
	\begin{definition}[Mixed Moves Distance : $d_N$]
		Given any two timing functions $\tau_1, \tau_2$ over the same causal process $(CN, p)$, we define the mixed move distance $d_N$ as follows 
		\small	$$d_N(\tau_1, \tau_2) = \min \{cost(m) | {m \in (Stamp \cup Delay)^*}, {m(\tau_1) = \tau_2}\}$$
		
	\end{definition}

	\begin{example}\label{stamptoy}
		Let us try to align the observed trace $(0, 3, 4)$ to the process trace $(0.5, 2.5, 3.5)$
		
		The best $d_t$ alignment for the example is a cost 1.5, with stamp moves editing each position. 
		
		The best $d_\theta$ alignment for the above pair is cost 1.5, as it requires a 0.5 delay edit at the first place to push the 0 forward to a 0.5, and then another 1 delay at the second position to pull the rest of the trace from $(3,4)$ back to $(2,3)$. 
		
		Now we consider the mixed setting. 
		Intuitively, when dealing with linear models where the ripple effect of a delay is always to the right irrespective of the values of the timestamps, moves can be applied in any order, so we can assume a minimum cost run is chronological. 
		So our first move will have to incur a minimum cost of $0.5$ as we try to align $(0,3,4)$ to our process trace at the first position.
		Now, if any part of this initial move was a delay, we would push the second component even further away from $2.5$ than it already was, and if any of the delay were negative, then the mixed move seems counterproductive at the first component as the delay and stamp moves work against each other, so let us say the first move was pure stamp. 
		Now we move on to move two, which will again incur a minimum cost of $0.5$. 
		Now, if the stamp part of this move were positive, it would leave the third component further off from the goal $3.5$ than if all of the move was delay, and so a pure delay move seems best here. 
		Hence, the least cost alignment is $(0.5, 0, 1)(0, 0.5, 2)$, with total cost 1. 
		
		All the intuitive reasoning given above will be justified in subsequent sections, and in fact distance 1 is the best we can do here even in the mixed case. 
	\end{example}

	\section{Computing $d_N$ on traces}\label{3}
	\subsection{Notation and Setup}
	Given the nonconstructive nature of the definition of $d_N$, it is not clear how one can efficiently calculate the distance between two fixed timed traces, as a minimal cost sequence of moves is not obvious. 
	As before, we first study the problem over sequential causal processes, as these make analysing the effects of moves significantly simpler. 
	Before we propose an algorithm that does calculate minimal cost for linear timed traces, we first define a few properties that seem to characterise classes of well formed minimal cost runs, and then prove that these properties both improve cost, and are satisfied solely by the run calculated by the algorithm we provide. 
	
	We start with some convenient notation for a common combination of the previously defined moves. 
	
	\begin{definition}[Mixed Move]
		We define \textit{mixed moves}, that denote doing a stamp move and a delay move at the same position in the word. 
		We define their effect as : $$(s, d, e)(\gamma) = stamp(delay(\gamma, d, e), s, e)$$
		Let the set of all mixed moves be $Moves$. 
	\end{definition}
	
	Given any move $m \in Moves$ we define the function $cost : Moves \to \mathbb{R}^+$ that returns the cost of the move.  
	The cost of a mixed move $(s, d, e)$ is the sum of the cost of the stamp and delay moves it's made up of, $|s| + |d|$.
	
	We say a sequence of moves $m_1m_2\dots m_n = m \in Moves^*$ \textit{aligns} $\tau_1$ to $\tau_2$ if $$m(\tau_1) = m_n(\dots (m_2(m_1(\tau_1))\dots) = \tau_2$$
	
Now, we look at a new way to represent timing functions. With delay edits, which constitute a new way of thinking about how a time-series can be transformed, comes a new perspective with which we can view timing functions over causal processes. 
	There is of course the standard definition, $\tau : \{1, \dots n\} \to \mathbb{R}^+$ that assigns to each event a timestamp that records exactly when the event occurs. 
	
	Instead, thinking along the lines of delays and durations between events occuring, we will often benefit from considering the following representation when speaking about delay moves, and so we define a way by which to view a timed word not in terms of its absolute timestamps, but by the delays between relevant timestamps. 
	
	\begin{definition}[Flow Function]
		Given a (not necessarily valid) time sequence,  $\tau : \{1, \dots n\} \to \mathbb{R}^+$, we first define the \textit{flow} function of $\tau$, $f_{\tau} : \{1, \dots n\} \to \mathbb{R}^+$ such that 
		$$f_{\tau}(i) = \begin{cases} 
			\tau(i) & i = 1\\
			\tau(i) - \tau(i-1) &i > 1\\
		\end{cases}
		$$
		
	\end{definition}
	
	The flow function thereby is a dual representation of timing functions, much like the relationship between graphs and line graphs, as if a timing function traditionally labels its transition nodes with timestamps, the flow function labels edges leading up to transitions with the duration since said transition was enabled. 
	
	\begin{example}\label{wordandflow}
		
%
%
		Consider the word $w = (abc)(1, 5, 9)$. For this linear timed trace, the flow function measuring its successive delays is $(1, 4, 4)$. 
	\end{example}
	
	As defined, $f_\tau$ as defined produces exactly the time durations that the guards of each transition in the model checks, i.e. the clock function values during the run.
	Hence, we see that if a word is in the language of the model, then its $f$ function maps events to values that lie within the constraint that the event's corresponding transition demands, that is, it is perfectly aligned with the model with distance 0. 
	This condition is unfortunately not sufficient though, due to urgency concerns, but it is quite close to the exact condition necessary for a word to be in the language of a time Petri net. 
	
	Also note that given the underlying causal process and the resulting $f_\tau$, we can reconstruct $\tau$ quite straightforwardly as $$\forall i \leq n : \tau(i) = \sum_{j \leq i} f_\tau(j)$$
	
	With a new representation, it is natural to study how the transformations we have defined earlier on timing functions transform their respective flow functions. A stamp move at an event $e$ of value $+x$ will increase the duration between $e$ and its immediate enabling predecessor, while shrinking the duration between $e$ and its immediate enabled successor by an equal amount, thereby moving $e$ by $x$ and compensating appropriately so none of the future timestamps change at all. 
	On the other hand, a delay move at $e$ will increase the duration between $e$ and its immediate enabling predecessor, keeping all other durations constant, thereby ensuring $e$ and all of its causal descendents have their timestamps shifted by $x$. 
	So both these moves are much more local as seen by the flow function, effectively perturbing only the duration attached to $e$ or its immediate enabled successor, while in the timing function representation moves at a position could have consequences well beyond its immediate neighbours. 
	This observation is summarised by the following lemma : 
	
	\begin{lemma}\label{FlowFuncTrans}
		Given a flow function $f_{\tau_1}$ and a mixed move $(s, d, i)$, the consequence of performing the move is the flow function defined below, where $\tau_2 = (s, d, i)\tau_1$. 
		$$f_{\tau_2}(j) = \begin{cases} 
			f_{\tau_1}(j) + s + d & j = i\\
			f_{\tau_1}(j) - s  & j = i-1\\
		\end{cases}
		$$
		
	\end{lemma}

	In order to compute the distance $d_N$ between two timed traces, we wish to cull down the space of possible sequences of moves that transform the two words to each other to smaller, equivalent possibility spaces that still capture all the behaviours we're interested in, while being a little more well behaved. 
	Hence, we introduce the following properties, which are designed intuitively to be sequences of moves that still perform the same effective transformations, and generally lower the cost of the transformation too by choosing moves wisely to avoid inefficiency. 
	These properties (order, co-operation, and stability) have been chosen specifically to uniquely characterise the runs of the algorithm we devise to compute $d_N$ on pairs of linear traces. 
	
	\subsection{Chronology}
	The first property captures the order in which moves in a sequence target the positions in a linear trace. 
	It would be good to be able to focus our attention on sequences of moves that proceed in an orderly fashion from one end of the word to another, doing exactly one mixed move at each position. This is intuitively reasonable as doing multiple mixed moves at the same position can only worsen the cost if it at all changes it, and assuming that mixed moves commute on linear words, it would be all the same to arrange them in order of the position it takes effect at. This gives us the following two properties. 
	
	\begin{definition}[Chronology]
		A sequence of moves aligning $\gamma$ to $\sigma$ is said to be \textit{chronological} if for all positions $i < j \leq n = |\gamma|$, all the moves at position $i$ are performed before any move at position $j$ and exactly one mixed move (where one or both components may be zero) takes place at each position, that is $\rho \in Moves^*$ is chronological iff  $$\forall i \in \{1, 2, \dots n\},  \exists s_i, d_i \in \mathbb{R} :$$$$ \rho = (s_1, d_1, 1)(s_2, d_2, 2)\dots (s_n, d_n, n) $$
		Such that $\forall i \in \{1, 2, \dots n\} : \sigma_i = \gamma_i + s_i + \sum_{j=1}^{i} d_j$. 
		
	\end{definition}
	
	\begin{definition}[Reverse Chronology]
		A sequence of moves aligning $\gamma$ to $\sigma$ is said to be \textit{reverse chronological} if for all positions $i < j \leq n = |\gamma|$, all the moves at position $i$ are performed after any move at position $j$ and exactly one mixed move (where one or both components may be zero) takes place at each position, that is $\rho \in Moves^*$ is chronological iff  $$\forall i \in \{1, 2, \dots n\},  \exists s_i, d_i \in \mathbb{R} : $$$$ \rho = (s_n, d_n, n)(s_{n-1}, d_{n-1}, n-1)\dots (s_1, d_1, 1) $$
		Such that $\forall i \in \{1, 2, \dots n\} : \sigma_i = \gamma_i + s_i + \sum_{j=1}^{i} d_j$. 
	\end{definition}
	
	Clearly, all chronological sequences have a one to one correspondence with reverse chronological sequences with exactly the same cost, the map between the two being simply reversing the sequence of moves. 
	In addition, in the flow function representation, the notions of chronology and reverse chronology are still somewhat preserved. 
	A mixed move affects both the duration before and all those after the event the move is applied at, so the order in which they happen does indeed stay chronological or reverse chronological as the case may be. 
	The only point of nuance here is that of course, a prefix of a sequence of say chronological moves no longer has the property of having completely aligned the word up to the i-th element of the flow vector while leaving the rest unchanged, as the last edit move may have perturbed the (i+1)st element. 
	
	But the upshot is, in a reverse chronological sequence of moves, the prefix of the sequence that stops at the ith event, combined with the stamp component of the mixed move at position (i-1), completely aligns the flow function for the suffix starting at position i. 
	Reverse chronological runs do not have an equivalent partial completion property in the timing function setting, but due to how local the moves end up being in the flow function setting, we obtain it here. 
	
	\textbf{Note : } For simplicity, we always assume the move played on the last position of a trace henceforth is a pure delay. Even if part of it were stamp, it would not affect the flow or timing representation any differently, so this is safe to assume. 
	
	\begin{example}
		
		Given $\gamma = (1, 1, 2, 4, 5)$ and $\sigma = (1, 2, 2.5, 4.2, 5)$, an example of a cost 5.3 non-chronological sequence of moves aligning them is as shown below. 
		$$\gamma \xmapsto{(-1, 0, 1)} (0, 1, 2, 4, 5) \xmapsto{(0, 2, 1)} (2, 3, 4, 6, 7) \xmapsto{(0, -1, 1)} $$
		$$(1, 2, 3, 5, 6) \xmapsto{(0.3, -0.8, 3)} (1, 2, 2.5, 4.2, 5.2) \xmapsto{(0, -0.2, 5)} \sigma$$
		There is an improved cost 3.3 chronological sequence of moves obtained by reordering the above sequence and combining moves at the same position : 
		$$\gamma \xmapsto{(-1, 1, 1)} (1, 2, 3, 5, 6) \xmapsto{(0, 0, 2)} (1, 2, 3, 5, 6) \xmapsto{(0.3, -0.8, 3)}  $$
		$$ (1, 2, 2.5, 4.2, 5.2) \xmapsto{(0, 0, 4)} (1, 2, 2.5, 4.2, 5.2) \xmapsto{(0, -0.2, 5)} \sigma $$
		
		There corresponding reverse chronological sequence of moves is of course : 
		$$\gamma \xmapsto{(0, -0.2, 5)} (1, 1, 2, 4, 4.8) \xmapsto{(0, 0, 4)} (1, 1, 2, 4, 4.8) \xmapsto{(0.3, -0.8, 3)} $$
		$$(1, 1, 1.5, 3.2, 4)\xmapsto{(0, 0, 2)} (1, 1, 1.5, 3.2, 4)\xmapsto{(-1, 1, 1)} \sigma $$
		
	\end{example}
	
	As the above example suggests, we claim : 
	
	\begin{lemma}\label{linescommute} There is a minimal cost sequence of moves aligning any two linear timed traces that is chronological/reverse chronological. 
	\end{lemma}
	
	\subsection{Co-operation}
	
	\begin{definition}[Co-operation]
		A mixed move is said to be co-operative if its delay edit and stamp edit are in the same direction, that is, 
		$$(s, d, i) \textrm{ is co-operative iff } sd \geq 0$$
		
		A chronological (or reverse chronological) sequence of moves $m = m_0m_1\dots m_{n-1}$ aligning $\gamma$ to $\sigma$ is said to be \textit{co-operative} if each of its moves is co-operative. 
	\end{definition}
	
	This property has nothing to do with the words it is trying to align, and so, is largely unaffected by whether we're viewing the words in timing or flow representation. 
	
	\begin{example}The previous example (that was a cost 3.3 run, recall) can hence be further improved to become co-operative cost 1.7 run as follows : 
		$$\gamma \xmapsto{(0, 0, 1)} (1, 1, 2, 4, 5) \xmapsto{(1, 0, 2)} (1, 2, 2, 4, 5) \xmapsto{(0.5, 0, 3)} $$
		$$(1, 2, 2.5, 4, 5) \xmapsto{(0.2, 0, 4)} (1, 2, 2.5, 4.2, 5) \xmapsto{(0, 0, 5)} \sigma $$
		
		With a reverse chronological counterpart on the flow vector: 
		$$f_\gamma \xmapsto{(0, 0, 5)} (1, 0, 1, 2, 1) \xmapsto{(0.2, 0, 4)} (1, 0, 1, 2.2, 0.8) \xmapsto{(0.5, 0, 3)} $$$$(1, 0, 1.5, 1.7, 0.8) \xmapsto{(1, 0, 2)} (1, 1, 1.5, 1.7, 0.8) \xmapsto{(0, 0, 1)} f_\sigma $$
		
	\end{example}
	
	Once again, as the example suggests, one can always convert a non-cooperative run into an equivalent cooperative one with cheaper cost. 
	
	\begin{lemma}\label{linescoop} There is a cooperative minimal cost chronological (or reverse chronological) sequence of moves aligning any two linear timed traces. 
	\end{lemma}
	
	Very similar to co-operation is the notion of \emph{cross co-operation}. Much like how mixed moves at a position should work together instead of against each other, a stamp move $s$ at a location should ideally have opposite sign to the delay $d$ at the next position, as its contribution to the next flow component is going to be $-s$. This leads us to the following definition : 
	
	\begin{definition}[Cross Co-operation]
		A reverse chronological run ${m = (0, d_n, n) \dots (s_1, d_1, 1)}$ is said to be cross co-operative if no stamp component has the same sign as the delay component of the move played at the next position, that is, $$\forall i < n : s_i \cdot d_{i+1} \leq 0$$
	\end{definition}
	
	Much like co-operation above, cross co-operation can be used to improve the cost of a run, as seen in the next lemma. 
	
	\begin{lemma}\label{crosscoop}
		For every reverse chronological co-operative run $m'$ aligning $\gamma$ to $\sigma$, there is a reverse chronological, co-operative and cross co-operative run $m$ aligning $\gamma$ to $\sigma$ such that $$cost(m) \leq cost(m')$$
	\end{lemma}
	
	\subsection{Stability}
	
	The final property we define will be largely used for reverse chronological runs on flow function representations, and it is the last improvement one can make on the cost of a sequence of moves as after enforcing this property, the sequence of moves that results is unique. 
	We prefer the flow representation for our current needs because when dealing with reverse chronological models, the flow function still retains the notion of partially aligning suffixes of the traces as we go along, provides a very useful constraint. 
	Moreover, we define stability in terms of flow values rather than timestamp values because the effects of mixed moves on flow vectors can be studied locally, while timestamp series get perturbed throughout their length by individual moves, and hence trying to optimise for the effect of a single mixed move involves too many variables. 
	Hence, until further notice, consider any co-operative run to be reverse chronological, and definitions and algorithms will favour aligning flow functions rather than timing functions. 
	It is quite clear to see that completely aligning a flow function is equivalent to completely aligning its timing function, so this is acceptable. 
	
	The notion of co-operation and reverse chronology leaves only one degree of freedom so to speak in our choice of aligning moves, that is, the ratio of stamp to delay at any given position. Once the stamp components of a sequence of moves has been decided, the delay components are predetermined to align the flow function from right to left. 
	Now, at any point, a delay move only affects one flow component, and adds to the cost of one mixed move.
	A stamp move also only adds to the cost of one mixed move, but it has the ability to affect two flow components, hence it is natural to want to choose stamp moves in a manner that corrects both the flow components it belongs to to the best of its abilities, thereby extracting as much use as possible from as little cost, and leaving the rest up to the delay component to correct. 
	It is this reasoning that brings us to the property we define below as stability, which denotes this prudent choice of $s$. 
	
	\begin{definition}[Stability]
		In a reverse chronological, co-operative sequence of moves, let  ${m_i = (s, d, i)}$ be a co-operative move seeking to correct the partially aligned flow function 
		
		$f_\gamma = (f_\gamma(1), \dots f_\gamma(i), f_\gamma(i+1), f_\sigma(i+2) \dots f_\sigma(n))$ (the result of the run of $m$ all the way up to and including the last stamp move) to $f_\sigma = (f_\sigma(1), f_\sigma(2), \dots f_\sigma(n))$. 
		
		Let $e_i = f_\sigma(i) - f_\gamma(i)$, and $e_{i+1} = f_\sigma(i+1) - f_\gamma(i+1)$. 
		
		We say $m_i$ is \textit{stable} if 
		$$s = \begin{cases}
			0 & e_i\cdot e_{i+1} \geq 0 \\
			e_i & \textrm{Else If }|e_i| < |e_{i+1}| \\
			-e_{i+1} & \textrm{Otherwise}
		\end{cases}$$
	\end{definition}
	
	A co-operative sequence of moves is said to be \textit{stable} if each of its moves is stable.
	Note that any stable run is also cross co-operative. 
	
		\begin{example}
		The earlier cost 1.7 example can be even further improved to be stable and have cost 1.5, thereby giving us the following :
		$$f_\gamma \xmapsto{(0, -0.2, 5)} (1, 0, 1, 2, 0.8) \xmapsto{(0, -0.3, 4)} (1, 0, 1, 1.7, 0.8)$$$$ \xmapsto{(0, 0, 3)} (1, 0, 1, 1.7, 0.8) \xmapsto{(0.5, 0.5, 2)} (1, 1, 1.5, 1.7, 0.8) \xmapsto{(0, 0, 1)} \sigma $$
		
	\end{example}
	
	Once again, as the example suggests, a stable chronological run (or reverse chronological run) is unique , but moreover, we claim that stability improves the cost of the sequence of moves.  
	
	\begin{lemma}\label{stableunique} The stable, co-operative, reverse chronological sequence of moves aligning any two linear timed traces is unique. 
	\end{lemma}
	
	\begin{proof}
		This is clear, as given a partially aligned flow trace the choice of $s$ to make the next move stable as defined above is deterministic, and the corresponding delay moves are determined by reverse chronology and co-operation. 
	\end{proof}
	
	\begin{lemma}\label{stablebest} There is a stable minimal cost sequence of moves aligning any two linear timed traces. 
	\end{lemma}
	
	\subsection{Computing $d_N$}
	
	\begin{figure}[!t]
		\removelatexerror
		\begin{algorithm}[H]
			\caption{$d_N$ Computation Algorithm}
			\label{MixBackwardsAlgo}
			\begin{algorithmic}
				\State \textbf{Input : }  $\sigma$, $\gamma$
				\State \textbf{Output : }  $d_N(\gamma, \sigma)$
				\State $cost \gets 0$
				\State $i \gets n$
				\While {$i > 1$}
				\State $a \gets f_\sigma(i) - f_\gamma(i)$
				\State $b \gets f_\sigma(i-1) - f_\gamma(i-1)$
				
				\If{$a\cdot b \geq 0$} 
				\Comment{$(0,a,i)$} 
			\State $\gamma \gets (0, a, i) \gamma$
			
			\ElsIf{$|a| < |b|$}
			\Comment{$(-a, 0, i-1)(0, 0, i)$} 
		
		\State $\gamma \gets (-a, 0, i-1)\gamma$
		\Else
		\Comment{$(b, 0, i-1)(0, a-b, i)$}
	\State $\gamma \gets (b, 0, i-1)(0, a-b, i) \gamma$
	\EndIf
	
	\State  $cost \gets cost + |a|$ 
	\State $i \gets i-1$
	\EndWhile
	
	\State $cost \gets cost + |\gamma_1 - \sigma_1|$  \Comment $(0, |\sigma_1- \gamma_1|, 1)$
	
\end{algorithmic}
\end{algorithm}
\end{figure}
\begin{lemma}\label {mixwordalgostable} Given two time sequences $\sigma = (\sigma_1, \sigma_2, \dots \sigma_n)$ and $\gamma = (\gamma_1, \gamma_2, \dots , \gamma_n)$ with linear underlying causal processes, the sequence of moves the above algorithm calculates ($m$) corresponds to the unique stable sequence of moves that aligns $\sigma$ to $\gamma$. 
\end{lemma}

\begin{proof}
This is clear by the definition of stability. 
\end{proof}

\begin{theorem}\label{MixAlgoCorr} Algorithm \ref{MixBackwardsAlgo} is correct, that is, its result \textsf(cost) $ = d_N(\gamma, \sigma)$
\end{theorem}

\begin{proof}
This holds as by Lemma \ref{mixwordalgostable} the algorithm calculates the unique stable sequence of aligning moves, and by Lemma \ref{stablebest} this must be the minimal cost sequence. 
\end{proof}

\section{Purely Timed Alignment for Linear Models}\label{4}

\begin{definition}[Sequential Process Models]
We define a sequential process model $N$ of length $n$ to be a sequence of intervals $\{[a_i, b_i] | a_i \in \mathbb{R}, b_i \in \mathbb{R} \cup \{\infty\}, i \leq n\}$. 

In addition, we define its language $\mathcal{L}(N)$ as follows : $$\{(t_1, \dots t_n) | \forall i \leq n : t_i - t_{i-1} \in [a_i, b_i]\}$$ Where $t_0 = 0$. 

We depict them as follows : 
$$\eye \longrightarrow \underset{[a_1, b_1]}{\square} \longrightarrow \bigcirc \longrightarrow \underset{[a_2,b_2]}{\square} \longrightarrow \dots \longrightarrow \underset{[a_n, b_n]}{\square} \longrightarrow \bigcirc$$
\end{definition}

\begin{example}\label{mixedmoveex}
Consider the below underlying sequential process model, and a new observed trace $\sigma = (3, 4, 5) \not \in \mathcal{L}(N)$. 
$$\eye \longrightarrow \underset{[0,1]}{\square} \longrightarrow \bigcirc \longrightarrow \underset{[2,2]}{\square} \longrightarrow \bigcirc \longrightarrow \underset{[1,1]}{\square} \longrightarrow \bigcirc$$
The best $d_t$ alignment for the example in the diagram below is $\gamma = (1,3,4)$ with minimum cost $d_t(\sigma, \gamma) = 4$. 

The best $d_t=\theta$ alignment for the example in the diagram below is also $\gamma = (1,3,4)$, but this time with minimum cost $d_\theta(\sigma, \gamma) = 3$, evidenced by the move sequence $(delay(-2, 1)delay(+1, 2))$. 

And lastly, the best $d_N$ alignment for the example in the diagram below is also $\gamma = (1,3,4)$, the sequence of moves being simply one stamp and one delay move at the start, $m = stamp(-1, 1)delay(-1, 1)$, and now with minimum cost $d_N(\sigma, \gamma) = 2 < \min \{d_t(\sigma, \gamma), d_\theta(\sigma, \gamma)\}$. 
\end{example}

\begin{theorem}\label{MixAlign} Given a sequential process model $N$ of a time Petri Net $N$ and a linear observed trace $\sigma$, the word $\gamma \in \mathcal{L}(N)$ such that $f_{\gamma}(i)= \underset{x \in [Eft(t_i), Lft(t_i)]} {\argmin}\hspace{0.5em} {|x - f_\sigma(i)|}$ also has the property $\gamma = \underset{x \in \mathcal{L}(N)}{\argmin} \hspace{0.5em}{d_N(x, \sigma)}$. 
\end{theorem}

This of course means that we have a linear time algorithm for aligning sequential processes under $d_N$, by locally choosing the best flow vector as shown in Algorithm \ref{MixAlignAlgo}. This, incidentally, is exactly the aligning word obtained by the algorithm developed for delay-only distance, $d_\theta$ in \cite{CR22}.  

\begin{figure}[!t]
\removelatexerror
\begin{algorithm}[H]
\caption{$d_N$ Alignment Algorithm}
\label{MixAlignAlgo}
\begin{algorithmic}
	\State \textbf{Input : }  $\sigma$, $N$
	\State \textbf{Output : }  $\gamma = \argmin_{x \in \mathcal{L}(N)} d_N(x, \sigma)$
	\For{$i <\in \{1, \dots, n\}$}
	\State $(a,b) \gets (Eft(t_i), Lft(t_i))$
	\State $f_{\gamma}(i)= \argmin_{x \in [a, b]} {|x - f_\sigma(i)|} $
	\State $i + + $
	\EndFor
\end{algorithmic}
\end{algorithm}
\end{figure}
\section{Implementation}

We have implemented the $d_N$ computation algorithm and the $d_N$ alignment algorithm for sequential process models in python, available at \url{https://github.com/NehaRino/TimedAlignments}. Both algorithms have linear time complexity, and so, they have efficient running times, as evidenced by the table below :\\
\centerline{
\begin{tabular}{|c| c|}
	\hline
	Trace Length & Running Time (seconds) \\ 
	\hline
	10 & 0.00003\\
	100 & 0.00024\\
	1000 &0.00259\\
	10000& 0.02811\\
	100000& 0.33131\\
	1000000& 3.44314\\
	\hline
\end{tabular}
}

\section{Perspectives and Conclusion}
In this paper, we studied the alignment problem for timed processes using the mixed moves distance $d_N$, and solved both distance computation and the purely timed alignment problem for the same. As far as we know, this (along with \cite{CR22}) is the first step in conformance checking for time-aware process mining, and much further work can be inspired from this point. Firstly, we have only solved the alignment problem for $d_N$ over sequential process models, which are rather structurally restricted and it would be interesting to see how to broaden the scope of these methods to larger classes of process models, such as branching process models or even general time Petri nets. Secondly, further investigation in the general timed alignment problem is necessary, as our proposed approach here is rather rudimentary and can certainly be improved. Lastly, there are a number of other conformance artefacts that can be set and studied in the timed setting, such as anti-alignments \cite{CBC21}, and it would be very interesting to better develop all such conformance checking methods in a manner that accounts for timed process models.

\bibliographystyle{IEEEtran}
\bibliography{Conformanceieee}

\begin{thebibliography}{10}
\providecommand{\url}[1]{#1}
\csname url@samestyle\endcsname
\providecommand{\newblock}{\relax}
\providecommand{\bibinfo}[2]{#2}
\providecommand{\BIBentrySTDinterwordspacing}{\spaceskip=0pt\relax}
\providecommand{\BIBentryALTinterwordstretchfactor}{4}
\providecommand{\BIBentryALTinterwordspacing}{\spaceskip=\fontdimen2\font plus
\BIBentryALTinterwordstretchfactor\fontdimen3\font minus
  \fontdimen4\font\relax}
\providecommand{\BIBforeignlanguage}[2]{{%
\expandafter\ifx\csname l@#1\endcsname\relax
\typeout{** WARNING: IEEEtran.bst: No hyphenation pattern has been}%
\typeout{** loaded for the language `#1'. Using the pattern for}%
\typeout{** the default language instead.}%
\else
\language=\csname l@#1\endcsname
\fi
#2}}
\providecommand{\BIBdecl}{\relax}
\BIBdecl

\bibitem{A16}
W.~M.~P. van~der Aalst, \emph{Process Mining - Data Science in Action, Second
  Edition}.\hskip 1em plus 0.5em minus 0.4em\relax Springer, 2016.

\bibitem{CDSW18}
\BIBentryALTinterwordspacing
J.~Carmona, B.~F. van Dongen, A.~Solti, and M.~Weidlich, \emph{Conformance
  Checking - Relating Processes and Models}.\hskip 1em plus 0.5em minus
  0.4em\relax Springer, 2018. [Online]. Available:
  \url{https://doi.org/10.1007/978-3-319-99414-7}
\BIBentrySTDinterwordspacing

\bibitem{Adriansyah2014AligningOA}
A.~Adriansyah, ``Aligning observed and modeled behavior,'' Ph.D. dissertation,
  Technische Universiteit Eindhoven, 2014.

\bibitem{BCC21}
\BIBentryALTinterwordspacing
M.~Boltenhagen, T.~Chatain, and J.~Carmona, ``A discounted cost function for
  fast alignments of business processes,'' in \emph{{BPM} 2021, Proceedings},
  ser. LNCS, vol. 12875.\hskip 1em plus 0.5em minus 0.4em\relax Springer, 2021,
  pp. 252--269. [Online]. Available:
  \url{https://doi.org/10.1007/978-3-030-85469-0\_17}
\BIBentrySTDinterwordspacing

\bibitem{Cheikhrouhou2014TheTP}
S.~Cheikhrouhou, S.~Kallel, N.~Guermouche, and M.~Jmaiel, ``The temporal
  perspective in business process modeling: a survey and research challenges,''
  \emph{Service Oriented Computing and Applications}, vol.~9, pp. 75--85, 2014.

\bibitem{Eder1999TimeCI}
J.~Eder, E.~Panagos, and M.~Rabinovich, ``Time constraints in workflow
  systems,'' in \emph{CAiSE}, 1999.

\bibitem{inbook}
A.~Nguyen, S.~Chatterjee, S.~Weinzierl, L.~Schwinn, M.~Matzner, and
  B.~Eskofier, \emph{Time Matters: Time-Aware LSTMs for Predictive Business
  Process Monitoring}, 03 2021, pp. 112--123.

\bibitem{RMAW13}
\BIBentryALTinterwordspacing
A.~Rogge{-}Solti, R.~Mans, W.~M.~P. van~der Aalst, and M.~Weske, ``Repairing
  event logs using timed process models,'' in \emph{On the Move to Meaningful
  Internet Systems: {OTM} 2013, Proceedings}, ser. LNCS, vol. 8186.\hskip 1em
  plus 0.5em minus 0.4em\relax Springer, 2013, pp. 705--708. [Online].
  Available: \url{https://doi.org/10.1007/978-3-642-41033-8\_89}
\BIBentrySTDinterwordspacing

\bibitem{CRHA20}
\BIBentryALTinterwordspacing
R.~Conforti, M.~L. Rosa, A.~H.~M. ter Hofstede, and A.~Augusto, ``Automatic
  repair of same-timestamp errors in business process event logs,'' in
  \emph{{BPM} 2020, Proceedings}, ser. LNCS, vol. 12168.\hskip 1em plus 0.5em
  minus 0.4em\relax Springer, 2020, pp. 327--345. [Online]. Available:
  \url{https://doi.org/10.1007/978-3-030-58666-9\_19}
\BIBentrySTDinterwordspacing

\bibitem{AS21}
\BIBentryALTinterwordspacing
W.~M.~P. van~der Aalst and L.~F.~R. Santos, ``May {I} take your order? - on the
  interplay between time and order in process mining,'' in \emph{Business
  Process Management Workshops - {BPM} 2021 International Workshops}, ser.
  Lecture Notes in Business Information Processing, A.~Marrella and B.~Weber,
  Eds., vol. 436.\hskip 1em plus 0.5em minus 0.4em\relax Springer, 2021, pp.
  99--110. [Online]. Available:
  \url{https://doi.org/10.1007/978-3-030-94343-1\_8}
\BIBentrySTDinterwordspacing

\bibitem{SSA11}
W.~Aalst, H.~Schonenberg, and M.~Song, ``Time prediction based on process
  mining,'' \emph{Inf. Syst.}, vol.~36, pp. 450--475, 04 2011.

\bibitem{CR22}
T.~Chatain and N.~Rino, ``{Timed Alignments},'' in \emph{4th International
  Conference on Process Mining, {ICPM} 2022}.\hskip 1em plus 0.5em minus
  0.4em\relax {IEEE}, 2022.

\bibitem{CBC21}
\BIBentryALTinterwordspacing
T.~Chatain, M.~Boltenhagen, and J.~Carmona, ``Anti-alignments - measuring the
  precision of process models and event logs,'' \emph{Inf. Syst.}, vol.~98, p.
  101708, 2021. [Online]. Available:
  \url{https://doi.org/10.1016/j.is.2020.101708}
\BIBentrySTDinterwordspacing

\end{thebibliography}
\clearpage
\appendix
We will detail the proofs of the theorems and lemmas stated above. 

\subsection{Chronology}
We start with the proof of Lemma \ref{linescommute}. 

\begin{lemma*}There is a minimal cost sequence of moves aligning any two linear timed traces that is chronological/reverse chronological. 
\end{lemma*}

\begin{proof}
We assume that there is at least one mixed move at each position, adding zero moves as needed. We first observe the following: 
$$\forall  i < j : (s, d, i)(s', d', j)(\gamma) = \gamma' = (s', d', j)(s, d, i)(\gamma)$$
Where $$\forall 1 \leq k \leq n : \gamma'_k = \begin{cases}
\gamma_k & k <i \\
\gamma_k + s + d & k = i \\
\gamma_k + d & i < k < j \\
\gamma_k + s' + d + d' & k = j \\
\gamma_k + d + d' & k > j
\end{cases}$$
That is, on linear timed traces, mixed moves are commutative. 
Clearly permuting the moves of a run does not change its cost, so we can assume any sequence of moves to be in increasing order of positions without loss of generality.

Now we can write $m'$ as follows : 
$$m' = (s_{11}, d_{11}, 1)\dots (s_{ij}, d_{ij}, i)\dots(s_{nk_n}, d_{nk_n}, n)$$ 

Where for each $i$, $k_i \geq 1$ is the number of mixed moves at position $i$. 
$$cost(m') = \sum_{i = i}^{n} (\sum_{j = 1}^{k_i} |s_{ij}| + |d_{ij}|)$$ 

Define $s_i = \sum_{j = 1}^{k_i} s_{ij}$, $d_i =\sum_{j = 1}^{k_i} d_{ij}$, giving us the chronological run $$m = (s_1, d_1, 1)\dots (s_i, d_i, i)\dots (s_n, d_n, n)$$

Now, by the triangle inequality we know that : 
$$\forall i : \sum_{j = 1}^{k_i} |s_{ij}| + |d_{ij}| \geq | \sum_{j = 1}^{k_i} s_{ij} | + | \sum_{j = 1}^{k_i} d_{ij} | = cost(s_i, d_i, i) $$

Thereby giving us that $$cost(m) \leq cost(m')$$ 

As for the reverse chronological sequence, by commutativity we can simply reverse $m$ and obtain $m_R$, a reverse chronological sequence of moves that does the same transformation and as reversing the sequence preserves cost, $$cost(m_R) = cost(m) \leq cost(m')$$ 
\end{proof}

\subsection{Co-operation} 
Now we prove Lemma \ref{linescoop}
\begin{lemma*}There is a cooperative minimal cost chronological (or reverse chronological) sequence of moves aligning any two linear timed traces. 
\end{lemma*}

\begin{proof}
We first prove that there is a co-operative and chronological minimal cost sequence of moves aligning any two linear timed traces. 
Given any minimal cost sequence of moves aligning $\gamma$ to $\sigma$, by Lemma \ref{linescommute} we know that there is a chronological sequence that is also minimal cost, call it $m'$. 

Assume $|\gamma| = |\sigma| = n$. 

We proceed by induction on the first index $k$ at which $m' = (s_1, d_1, 1)\dots(s_n, d_n, n)$ uses a non co-operative move. 

\textbf{Base case} : $k = n$.

This gives us $$s_n d_n < 0 \implies |s_n + d_n| < |s_n| + |d_n|$$ 
$$\implies m = m'|_{n-1}(0, s_n + d_n, n)$$ costs less than $m'$, aligns the words and is co-operative. 

\textbf{Induction hypothesis} : Suppose we've demonstrated that there is a sequence of moves $m''$ that aligns $\gamma$ to $\sigma$, is co-operative for the first $k-1 < n-1$ steps, is non co-operative at step $k$ and $cost(m'') \leq cost(m)$

\textbf{Claim} : There is a sequence of moves $m$ that aligns $\gamma$ to $\sigma$, is co-operative for at least the first $k$ steps and $$cost(m) \leq cost(m'') \leq cost(m')$$

\textbf{Proof} : Let $m'' = (s_1, d_1, 1)\dots(s_n, d_n, n)$. 

We first define $m_i = m''|_{k-1}$, $m_f = (s_{k_2}, d_{k+2}, k+2)\dots (s_n, d_n, n)$, thereby giving us $$m'' = m_i(s_k, d_k, k)(s_{k+1}, d_{k+1}, k+1)m_f$$

We know that $s_kd_k < 0$. 
This can be split into the following cases : 

\textbf{Case 1} : $|d_k + s_k| = |d_k| - |s_k|$

We define the following aligning, chronological run that is co-operative up to (and including) the $k$th step : 
$$m = m_i( 0, d_k + s_k, k)( s_{k+1}, d_{k+1} - s_k, k+1)m_f$$
$$cost(m'') - cost(m)  $$$$ = |s_{k+1}| + |d_{k+1}| + |s_k| + |d_k| - ( | d_k + s_k| + | s_{k+1}| + | d_{k+1} - s_k |) $$ $$\geq (|s_{k+1}| + |d_{k+1}|  - |s_{k+1}| + | d_{k+1}|  )+ (  |s_k| + |d_k|  -| d_k + s_k|  -|s_k |)   \geq 0$$

\textbf{Case 2} : $|d_k + s_k| = |s_k| - |d_k|$

We define the following aligning, chronological run that is co-operative up to (and including) the $k$th step : 
$$m = m_i( d_k + s_k, 0, k)( s_{k+1}, d_k + d_{k+1}, k+1)m_f$$
$$cost(m'') - cost(m) $$$$= |s_{k+1}| + |d_{k+1}| + |s_k| + |d_k| - (|s_k + d_k | + | s_{k+1}| + | d_{k+1} + d_k |) $$
$$\geq |d_{k+1}| + |s_k| + |d_k| - (|s_k + d_k | + | d_{k+1}| + |d_k |) \geq 0$$

By induction we can hence conclude that there is a chronological sequence of moves $m$ that is co-operative throughout and has lower cost than $m'$, i.e., if $m'$ is minimal cost so is $m$.  

Now, $m_R$ is also co-operative throughout as co-operation is a property of the individual moves constituting the sequence, and $m_R$ has exactly the same moves just in the reverse order, so there is also a reverse chronological minimal cost sequence of moves aligning $\gamma$ to $\sigma$. 

\end{proof}

We proceed to the proof of Lemma \ref{crosscoop}. 

\begin{lemma*}	For every reverse chronological co-operative run $m'$ aligning $\gamma$ to $\sigma$, there is a reverse chronological, co-operative and cross co-operative run $m$ aligning $\gamma$ to $\sigma$ such that $$cost(m) \leq cost(m')$$
\end{lemma*}

\begin{proof}
Say $m' = (0, d'_n, n) \dots (s'_1, d'_1, 1)$. 
Let $i$ be an index at which the stamp move is not cross co-operative, i.e., $s_i \cdot d_{i+1} > 0$. 

Consider a run $m''$ that is identical to $m'$, except for the moves it plays at positions $i$ and $i+1$.   

Let $m''$ perform the move $(0, d_i + s_i, i)$ at position $i$ and $(s_{i+1}, d_{i+1} - s_{i}, i+1)$ for the move at position $i+1$ if $|s_i| < |d_i|$ and $(s_i - d_{i-1}, d_i + d_{i+1}, i)$ at position $i$ and $(s_{i+1}, 0, i+1)$ for the move at position $i+1$ otherwise. 

Clearly, the appropriate one of either of these is cross co-operative at position $i$, while also staying co-operative and reverse chronologically aligning $\gamma$ to $\sigma$, and costs $\min(|s_i|, |d_{i-1}|)$ less that $m'$. 

Doing this procedure for at most $n$ times will result in a run with every stamp move being cross co-operative and with lower cost than $m'$, that is, our required $m$. 
\end{proof}

\subsection{Stability} 
Before we proceed with the proof of this result, we prove a small technical lemma that will help us in the proof of Lemma \ref{stablebest}. The motivation behind the following result is essentially that the mixed distance $d_N$ respect an intuitive notion of closeness, once two words have been aligned up till the penultimate place, the one whose last duration since firing lands it closer to the desired target will naturally be the easier one to align. Put another way, if two words, aligning to a target word, are identical to each other in all respects but the last timestamp, then the one whose last flow value is closer to the desired target flow value, is the word that is over closer to the target under $d_N$ between the two. 

\begin{lemma}\label{cooplast}
Given two timing functions $$\gamma^x = (\gamma_1, \gamma_2, \dots \gamma_{n-1}, \gamma_{n-1} + x)$$ $$\gamma^y = (\gamma_1, \gamma_2, \dots \gamma_{n-1}, \gamma_{n-1} + y)$$ Both aligning to $\sigma = (\sigma_1, \dots, \sigma_n)$, such that $|x - (\sigma_n - \sigma_{n-1})| \leq |y - (\sigma_n - \sigma_{n-1})|$ then $$d_N(\sigma, \gamma^x) \leq d_N(\sigma, \gamma^y)$$ 
\end{lemma}

\begin{proof}
Consider any run $m^y$ aligning $\gamma^y$ to $\sigma$. 
We claim there is a run $m^x$ aligning $\gamma^x$ to $\sigma$ such that $cost(m^x) \leq cost(m^y)$. 

This proves the above lemma as then the minimal cost run aligning $\gamma^x$ to $\sigma$ would have cost at most $m^x$ corresponding to the minimal cost run aligning $\gamma^y$ to $\sigma$, so $$d_N(\sigma, \gamma^x) \leq d_N(\sigma, \gamma^y)$$

Now, we can assume without loss of generality by Lemmas \ref{linescommute} and \ref{linescoop}  that $m^y$ is chronological and co-operative, say $m^y = (s_1, d_1, 1) \dots (0, d_n, n)$. 

We construct $m^x = (s_1, d_1, 1)(s_2, d_2, 2)\dots (s'_{n-1}, d'_{n-1}, n-1)(0, d'_n, n)$ differing at most at the last two mixed moves. 
As $m^x|_{n-2} = m^y|_{n-2}$ we can also conclude that $s_{n-1} + d_{n-1} = s'_{n-1} + d'_{n-1}$ by chronology, so with co-operation we can conclude that $cost(m^y) - cost(m^x) = |d_n| - |d'_n|$. 

Hence, the problem is reduced to proving that whatever the choice of $d_n$, there is a choice of $d'_n$ that fully aligns $\gamma^x$ and has atmost equal magnitude. 

Now, $|d_n| = |\sigma_{n-1} + y - \sigma_n|$ as $m^y|_{n-1}$ fully corrected up till the penultimate place, and similarly $|d'_n| = |\sigma_{n-1} + x - \sigma_n|$, hence, by definitions of $x$ and $y$, we have the result. 

\end{proof}

Now, we can proceed with the proof of Lemma \ref{stablebest}. 
\begin{lemma*}
	There is a stable, co-operative, reverse chronological minimal cost sequence of moves aligning any two linear timed traces. 
\end{lemma*}
\begin{proof}
Suppose we seek to align the linear timed trace $\gamma$ to $\sigma$, and consider a minimal cost, reverse chronological, co-operative and cross co-operative sequence of moves $m'$ that aligns $\gamma$ to $\sigma$. 
We know $m'$ exists by Lemmas \ref{linescoop} and \ref{crosscoop}. 
On the other hand, consider the unique stable, co-operative, reverse chronological sequence of moves $m$ that aligns $\gamma$ to $\sigma$, that exists by Lemma \ref{stableunique}. 

We proceed by way of induction. 
In the base case of words of length 1, the claim is clearly true as the minimal cost aligning move would just be the minimum delay move, which is also the unique stable run. 
For words of length 2, i.e., $\gamma_1\gamma_2$ aligning to $\sigma_1\sigma_2$, a reverse chronological co-operative run on the flow function is of the following form : $$m' : f_\gamma \xmapsto{(0,d_2)} (f_\gamma(1), f_\gamma(2) + d_2)$$$$ \xmapsto{(s_1,d_1)} (f_\gamma(1) + s_1 + d_1, f_\gamma(2) - s_1 + d_1) = f_\sigma$$

By reverse chronology and co-operation, we have the following equation : 
$$cost(m') = |d_2| + |s_1 + d_1| = |f_\sigma(2) - f_\gamma(2) + s_1| + |f_\sigma(1) - f_\gamma(1)|$$

Clearly the only variable in the cost is the value of $s_1$, and the stable run $m$ is the run that precisely minimises the value $|f_\sigma(2) - f_\gamma(2) + s_1|$ while maintaining co-operation, so the claim holds. 

Now, let the length of $\gamma$ be $n > 2$, and suppose the claim holds for all words of length less than $n$. 
We claim that $cost(m) \leq cost(m')$. 

The two runs must diverge at some point, as otherwise their cost is equal and the claim is obvious. There are only two main cases of interest : 
\begin{mycase}
\case The first move is the same, but the second moves are different (i.e the mixed move played at the second right-most position of the trace). 
\case The first moves of the two sequences $m$ and $m'$ (i.e the delay played on the right-most position of the trace) are different. 

\end{mycase}

If the first move where the sequences diverge are later, say at some position $i$ in the trace, then we can restrict our attention to the $i+1$ length prefix of the word and the appropriate truncations of $m$ and $m'$, and by induction hypothesis we are done. 
So, we focus on the above cases henceforth. 

\begin{mycase}
\case Let $m = (0, d_n, n)(s_{n-1}, d_{n-1}, n-1) \dots (s_1, d_1, 1)$ and 

$m' = (0, d'_n, n)(s'_{n-1}, d'_{n-1}, n-1) \dots (s'_1, d'_1, 1)$ 

As both these runs are reverse chronological, we obtain the following equalities: $$\forall i > 1 : s_i + d_i - s_{i-1} = f_\sigma(i) - f_\gamma(i) = s'_i + d'_i + s'_{i-1}$$ 
Setting $i$ to $n$ gives $d_n - s_{n-1} = d'_n - s'_{n-1} = f_\sigma(n) - f_\gamma(n)$, which implies that  since $d_n = d'_n$, $s_{n-1} = s'_{n-1}$. 

Now consider $\gamma_s = (s_{n-1}, 0, n-1)\gamma|_{n-1}$, and 

$m_s = (0, d_{n-1}, n-1)(s_{n-2}, d_{n-2}, n-2)\dots (s_1, d_1, 1)$ and 

$m'_s = (0, d'_{n-1}, n-1)(s'_{n-2}, d'_{n-2}, n-2)\dots (s_1, d_1, 1)$. 

Clearly $m_s$ is the stable run aligning $\gamma_s$ to $\sigma|_{n-1}$ and $m'_s$ is another reverse chronological co-operative aligning sequence, and $$cost(m) = |d_n| + |s_{n-1}| + cost(m_s)$$ $$cost(m') = |d'_n| + |s'_{n-1}| + cost(m'_s)|$$ and by triangle inequality and the stability of $m$, $$|d_n| + |s_{n-1}| = |d_n - s_{n-1}| = |d'_n - s'_{n-1}| \leq |d'_n| + |s_{n-1}|$$ 
Now, by induction hypothesis on $\gamma_s, m_s, m'_s$ (they're an instance of length $n-1$ words, case 2),  we deduce $cost(m_s) \leq cost(m'_s)$, by which we can conclude that $$cost(m) \leq cost(m')$$

\case Let $m = {(0, d_n, n)(s_{n-1}, d_{n-1}, n-1) \dots (s_1, d_1, 1)}$ and 
$m' = (0, d'_n, n)(s'_{n-1}, d'_{n-1}, n-1) \dots (s'_1, d'_1, 1)$ 

As both these runs are reverse chronological, we obtain the following equalities: $$\forall i > 1 : s_i + d_i - s_{i-1} = f_\sigma(i) - f_\gamma(i) = s'_i + d'_i + s'_{i-1}$$ 
Setting $i$ to $n$ gives $d_n - s_{n-1} = d'_n - s'_{n-1} = f_\sigma(n) - f_\gamma(n)$, which implies that having performed the first delay and stamp move of each run, one gets the following :  

$$f_\gamma \xmapsto{(0, d, n)(s, 0, n-1)} (f_\gamma(1), \dots f_\gamma(n-1) + s, f_\sigma(n))$$
$$f_\gamma \xmapsto{(0, d', n)(s', 0, n-1)} (f_\gamma(1), \dots f_\gamma(n-1) + s', f_\sigma(n))$$
Where $s = s_{n-1}$, $d = d_n$, $s' = s'_{n-1}$, and $d' = d'_n$. 

Now, let $e_n = f_\sigma(n) - f_\gamma(n)$, and $e_{n-1} = f_\sigma(n-1) - f_\gamma(n-1)$
Now, by stability of $m$, we have the following definition of $s$ : 
$$s = \begin{cases}
	0 & e_n\cdot e_{n-1} \geq 0 \\
	e_{n-1} & \textrm{Else If }|e_{n-1}| < |e_n| \\
	-e_n & \textrm{Otherwise}
\end{cases}$$

By the above definition, we see that $s = \argmin_{x \cdot (e_n - x) \geq 0} {|e_{n-1} - x|}$. 

Now by cross co-operation, we know that $s' \cdot (e_n - s') \geq 0$, so this means $$|e_{n-1} - s| \leq |e_{n-1} - s'| $$$$\implies |f_\sigma(n-1) - (f_\gamma(n-1) + s)| \leq |f_\sigma(n-1) - (f_\gamma(n-1) + s')|$$

Now, consider the words $\gamma_s = (\gamma_1, \dots \gamma_{n-2}, \gamma_{n-1} + s)$ and 
$\gamma'_s = (\gamma_1, \dots \gamma_{n-2}, \gamma_{n-1} + s')$ both aligning to $\sigma|_{n-1}$. We can apply Lemma \ref{cooplast} to these, and deduce that $$d_N(\gamma_s, \sigma|_{n-1}) \leq d_N(\gamma'_s, \sigma|_{n-1})$$

Now, by the induction hypothesis, as the rest of $m$ would have been the stable minimal cost run aligning $\gamma_s$ to $\sigma|_{n-1}$, along with stability of $m$, we have the following : 
$$cost(m) = |s| + |d| + d_N(\gamma_s, \sigma|_{n-1})$$$$ \leq |d-s| + d_N(\gamma'_s, \sigma|_{n-1}) \leq cost(m')$$

And hence we are done. 
\end{mycase}

\end{proof}

\subsection{Purely Timed Alignment for Sequential Process Models} 
We can now look at the proof of Theorem \ref{MixAlign}. 

\begin{theorem*}
	Given a sequential process model $(CN, p)$ of a time Petri Net $N$ and a linear observed trace $\sigma$, the word $\gamma \in \mathcal{L}(N)$ such that $f_{\gamma}(i)= \underset{x \in [Eft(t_i), Lft(t_i)]} {\argmin}\hspace{0.5em} {|x - f_\sigma(i)|}$ also has the property $\gamma = \underset{x \in \mathcal{L}(N)}{\argmin} \hspace{0.5em}{d_N(x, \sigma)}$. 
\end{theorem*}
\begin{proof}
As $(CN, p)$ is linear, we can consider the event set $E$ to be totally ordered by $G$, giving us the list $\{e_1, e_2, \dots e_n\}$ and their respective static interval constraints $\forall i \leq n : SI(p(e_i)) = [a_i, b_i]$. In this setting, clearly, $$f_{\gamma}(i)= \argmin_{x \in [a_i, b_i]} {|x - f_\sigma(i)|}$$

We claim that for all $\alpha \in \mathcal{L}(N)$ i.e $\forall \alpha : f_{\alpha}(i) \in [a_i, b_i]$, $$d_N(\sigma, \gamma) \leq d_N(\sigma, \alpha)$$

By Theorem \ref{MixAlgoCorr} we shall henceforth implicitly equate the cost of the stable run aligning a word to the mixed moves distance to the word. 
Let $m, m'$ be the stable runs aligning $\gamma, \alpha$ respectively to $\sigma$, where: 
$$m = (0, d_n, n)(s_{n-1}, d_{n-1}, n-1) \dots (s_1, d_1, 1)$$  
$$m' = (0, d'_n, n)(s'_{n-1}, d'_{n-1}, n-1) \dots (s'_1, d'_1, 1)$$ 

In addition, we know that $cost(m) = \sum_{i = 1}^{n} (|s_i| + |d_i|)$, and $cost(m') = \sum_{i = 1}^{n} (|s'_i| + |d'_i|)$. 

We break the run up after each stamp move, but not delay move, at a particular position has been performed, as below: $$\forall i > 0 : c_i = |s_{i-1} + |d_i|, c'_i = |s'_{i-1}| + |d'_i|$$ Where we let $s_0 = s'_0 = 0$ for notation's sake. Clearly $cost(m) = \sum_{i = 1}^{n} c_i$ and $cost(m') = \sum_{i = 1}^{n} c'_i$

We wish to prove that $$cost(m) \leq cost(m')$$ 

In fact, we shall prove something stronger. We define the notion of the disadvantage of $m$, which denotes in effect the amount of distance it might have to catch up in aligning due to $m'$ having done a costlier move just prior. $$dis(i) = max(0, c'_i - c_i)$$

Note of course, that each $c_i$, $c'_i$ and $dis(i)$ is nonnegative. 

Now, we claim the following : $$\forall k > 0 : \sum_{i = k}^{n} c_i + dis(k) \leq \sum_{i=k}^{n} c'_i$$
If we prove this, we will be done, as the equation for $k=1$ proves that $m$ has overall lower cost than $m'$. 

Let us proceed by induction on $n-k$. 

For the base case, $k = n$, so the equation we seek is that $c_n + dis(n) \leq c'_n$, which is clear as $dis(n) = max(0, c'_n - c_n)$. 

Now, say the claim holds for all $k > K$ for some $K < n$. We seek to now prove that $$c_K + \sum_{i = K+1}^{n} c_i + dis(K) \leq c'_K + \sum_{i = K+1}^{n} c'_i$$

Now, if $c_K \leq c'_K$ this is clear as $c_K + dis(K) = max(c_K, c'_K)$. On the other hand, suppose $c_K > c'_K$. 

By the induction hypothesis, we know that $$\sum_{i = K+1}^{n} c_i + dis(K+1) \leq \sum_{i = K+1}^{n} c'_i$$ So if we show that $c_K - c'_K \leq dis(K+1)$, we are done. 

Now, by reverse chronology, co-operation and cross co-operation of $m$ and $m'$, and the definition of $\gamma$, we know the following inequality holds : $$|s_{K}| + c_K = |s_K + d_K - s_{K-1}| = |f_\sigma(K) - f_\gamma(K)|$$ $$ \leq |f_\sigma(K) - f_\alpha(K)| =  |s'_K + d'_K - s'_{K-1}| = |s'_{K}| + c'_K$$

So, $c_K > c'_K \implies |s_K| < |s'_K|$, and moreover, $$c_K - c'_K \leq |s'_K - s_K|$$

Now, by stability of $m$, we know that $s_K$ was defined as below : 

$$s_K = \begin{cases}
0 & e_K \cdot e_{K+1} \geq 0 \\
e_K & \textrm{Else If } |e_K| < |e_{K+1}|\\
e_{K+1} & \textrm{Otherwise}
\end{cases}$$

Where $e_K =  (f_\sigma(K) - f_\gamma(K))$ and $e_{K+1} = f_\sigma(K+1) - f_\gamma(K+1) - s_{K+1}$. 

We proceed to analyse the possibilities based on the various possible stable choices of $s_K$ depending on the state of the run so far. 

Suppose firstly that $s_K = 0$. There are a number of ways this could hold. By choice of $\gamma$, $$(f_\sigma(K) - f_\gamma(K))\cdot (f_\sigma(K+1) - f_\gamma(K+1)) > 0 $$$$ \implies (f_\sigma(K) - f_\alpha(K))\cdot (f_\sigma(K+1) - f_\alpha(K+1)) > 0 \implies s'_K = 0$$

But $|s'_K| >|s_K| \geq 0 $ so this first subcase is not possible. 

Secondly, if $$(f_\sigma(K) - f_\gamma(K)) = 0 \implies c_K = 0 > c'_K$$ Which is also not possible. 

Thirdly, if $$(f_\sigma(K+1) - f_\gamma(K+1) - s_{K+1} = 0 \implies d_{K+1} = 0 $$

Now, $d_{K+1} = 0 $ and $|s'_K| > |s_K|$ implies $$dis(K+1) = |s'_K| + |d'_{K+1}| - |s_K| \geq |s'_K|  - |s_K| \geq c_K - c'_K$$ So in this case the claim holds as needed. 

Otherwise, suppose $s_K = f_\sigma(K) - f_\gamma(K)$. In this case, $c_K$ is once again equal to zero, which is not possible as by assumption $c_K > c'_K \geq 0$. 

Now, the last possibility is that $s_K = f_\gamma(K+1) - f_\sigma(K+1) -s_{K+1}$. This means that once $s_K$ is played, the $K+1$ position is perfectly aligned, so $d_{K+1} = 0$. Once again, this means $$dis(K+1) = |s'_K| + |d'_{K+1}| - |s_K| \geq |s'_K|  - |s_K| \geq c_K - c'_K$$ So in this case also, the claim holds as needed. 

Hence, we see that in every possible case, the claim holds for $k= K$ as well, so by induction we see that $$\sum_{i = 1}^{n} c_i + dis(K+1) \leq \sum_{i = 1}^{n} c'_i$$ From which we can conclude that $$cost(m) = \sum_{i = 1}^{n} c_i ) \leq \sum_{i = 1}^{n} c'_i = cost(m')$$

Hence, the result holds, and $\gamma$ is an optimal word of $\mathcal{L}(N)$ to align $\sigma$ to. 

\end{proof}
\end{document}